\documentclass[12pt]{article}
\usepackage{fullpage}
\usepackage[authoryear]{natbib}
\usepackage{amsmath,amssymb,amsfonts}
\usepackage{bbm,mathrsfs,multirow,bm}
\usepackage{graphicx}
\usepackage{enumitem}
\usepackage{float}
\usepackage{caption,subcaption}
\usepackage[ruled]{algorithm2e}
\usepackage{array,booktabs}
\usepackage{fontenc}
\usepackage[usenames,dvipsnames,svgnames,table]{xcolor}
\usepackage[colorlinks=true, linkcolor=red, urlcolor=blue, citecolor=blue]{hyperref}
\usepackage{setspace}

\newcommand{\bg}{\begin{eqnarray}}
\newcommand{\ed}{\end{eqnarray}}
\newcommand{\bgn}{\begin{eqnarray*}}
\newcommand{\egn}{\end{eqnarray*}}

\newcommand{\trans}{^{\scriptscriptstyle \mathsf T}}
\newcommand{\Beta}{\bm{\beta}}
\newcommand{\btheta}{\bm{\theta}}

\newcommand{\cov}{\operatorname{Cov}}

\newtheorem{theorem}{Theorem}[section]

\begin{document}

\begin{center}
{\Large \bf Semi-supervised Method for Risk Prediction with Doubly Censored EHR Data} \vspace{.2in}
\end{center}

\begin{center}\vspace{-2mm}
{Jie Zhou, Enhao Wang\\
School of Mathematics, Capital Normal University, Beijing 100048, China\bigskip \\
 Xuan Wang\footnote{Correspondence: xuan.wang@utah.edu}\\
 Division of Biostatistics, Department of Population Health Sciences, University of Utah, SLC, UT 84108, U.S.A.} \vspace{-2mm}
\end{center}

\begin{abstract}
The rapid expansion of large-scale electronic health record (EHR) data offers unique opportunities to improve the accuracy and efficiency of clinical risk estimation. Yet, because clinical events may occur outside the recording health system, clinical event outcomes are frequently subject to double censoring (both left and right). Besides, gold-standard event times can often only be ascertained through labor-intensive manual chart reviews, yielding labels for only a small subset of patients. Reliance on this limited labeled set alone is limited in efficiency, whereas widely available surrogate outcomes such as the time to first diagnostic code or first disease mention are error-prone and can yield biased estimates if used directly. Semi-supervised learning (SSL) methods provide a principled way to integrate labeled and unlabeled data, and prior work has demonstrated their advantages in settings with binary or right-censored outcomes. However, existing approaches do not accommodate double censoring for risk prediction, which poses additional methodological challenges. To address this gap, we develop a novel SSL framework for risk prediction that combines a small set of gold-standard labels with large-scale surrogate information under double censoring. We establish the theoretical validity of the proposed estimator. Through extensive simulation studies, we show that our method substantially improves estimation efficiency relative to existing supervised estimators (based on the labeled data). Finally, we demonstrate its practical value by applying it to study risk factors for type 2 diabetes (T2D) using EHR data from a health system in the US.\\

\end{abstract}

\newpage
\section{Introduction}
Electronic Health Records (EHRs) have emerged as a uniquely valuable resource in translational research. They furnish a wealth of clinical details through comprehensive patient profiles, encompassing a broad and heterogeneous patient population that is often underrepresented in Randomized Clinical Trials (RCTs), and enable access to large sample sizes with extensive longitudinal follow-up. These attributes render EHRs particularly potent in examining real-world care patterns and outcomes, thereby complementing the evidence generated by RCTs and prospective cohort studies. Simultaneously, employing EHR data to estimate treatment or risk effects poses significant methodological challenges.

The event time outcomes in EHR data are often affected by double censoring: the true onset time of a disease may precede the first recorded visit (left censoring), while the event of interest may remain unobserved at the end of follow-up (right censoring). This situation arises in our motivating example of using EHR data from a health system in the US for predictive modeling of type 2 diabetes mellitus (T2D) onset risk. In this cohort, some patients had T2D at their first visit (left censoring), others developed T2D during follow-up, and some remained diabetes-free at their last visit (right censoring). Ignoring double censoring can lead to biased estimates of survival functions, hazard ratios, and treatment effects, and poses challenges to identifiability under real-world EHR data structures. A substantial statistical literature exists for analyzing double censored data. Estimates of survival functions based on double censored data have been studied by many researchers \citep{turnbull1974nonparametric, tsai1985large, chang1990weak, gu1993asymptotic}. \citet{zhang1996linear} and \citet{ren1997regression} studied linear regression models with double censored data, \citet{kim2010asymptotic} and \citet{kim2013algorithm} investigated proportional hazards models, while \citet{cai2004semiparametric} and \citet{li2018class} considered semiparametric transformation models and adopted methods based on unbiased estimating equations, respectively, using both maximum likelihood and method-of-moments estimators.

Another key challenge lies in measuring the primary outcome of interest, such as the doubly censored time to onset of T2D in our example. Accurate determination often requires substantial manpower and time for medical record review \citep{zhang2019high}, which may not be feasible on a large scale and is typically only applicable to a limited subset of patients, known as the labeled set. The remaining patients are referred to as the unlabeled set. In contrast, surrogate outcomes, such as the time of the first International Classification of Diseases (ICD) code for T2D or the first mention of T2D in clinical records, are easily accessible across the entire cohort but are subject to substantial misclassification. Directly relying on such error-prone surrogate outcomes can lead to biased survival estimates, while using only a small labeled set can result in high variability. These limitations have spurred the development of semi-supervised methods, which utilize limited gold standard labels alongside large-scale surrogate outcome-based data to achieve more efficient and reliable inference. In the context of EHRs, semi-supervised learning (SSL) methods have been extensively studied, particularly in classification-based phenotypic analysis using machine learning approaches \citep{beaulieu2016semi, zhang2019high}. Beyond classification, \citet{chakrabortty2018efficient} and \citet{zhang2019semi} have developed SSL inference procedures for continuous outcome regression, while \citet{gronsbell2018semi} proposed an SSL method for evaluating the predictive performance of binary outcomes. Recently, \citet{ahuja2023semisupervised} introduced SSL estimation techniques for survival analysis with current status data. For right-censored data, enhanced estimation methods have been proposed to improve supervised estimators based solely on labeled data \citep{chen2002cox, jiang2007additive, yu2013adjusted, wang2015semiparametric, tong2020augmented}. However, to our knowledge, there is currently no SSL method for risk prediction with doubly censored data.

The rest of this paper is organized as follows. Section 2 introduces the proposed semi-supervised estimator, describing the steps based on supervised likelihood and the enhancement steps utilizing alternative results. Section 3 establishes the asymptotic properties of the estimator, including consistency, asymptotic normality, and efficiency improvement relative to supervised methods. Section 4 reports the results of extensive simulation studies, demonstrating the finite sample performance of the method. Section 5 illustrates the practical value of the method by applying the proposed approach to the double-censored EHR data of the health system regarding T2D, emphasizing its improved efficiency and robustness compared to standard supervised methods. Section 6 summarizes the significance, potential extensions, and limitations of the method. Technical proofs of asymptotic results and additional simulation findings are provided in the Appendix.

\section{Model and Estimation}

\subsection{Notation}

For individual $i$, let $T_i$ denote the event time of interest subject to double censoring. Let $L_i$ and $U_i$ denote the left and right censoring times, respectively, that are always observed and satisfy $P(L_i < U_i)=1$. Under double censoring, $T_i$ is observed only if it falls within the interval $[L_i, U_i]$, while $T_i$ is left-censored if $T_i < L_i$ and right-censored if $T_i > U_i$. Thus, the observed data for $T_i$ can be represented by $\{X_i=\max(L_i,\min(T_i,U_i)),\delta_i\}$, where $\delta_i=1$ if $X_i= T_i$; $\delta_i= 2$ if $X_i=U_i$; $\delta_i= 3$ if $X_i=L_i$. Suppose there is a vector of $q$ surrogate outcomes $\mathbf{T}_{i}^{\ast}=(T_{1i}^{\ast},\ldots,T_{qi}^{\ast})^{\trans}$, which are proxies of $T_i$ and also subject to double censoring. Similarly, for $\mathbf{T}_{i}^{\ast}$, we only observe $\mathbf{X}_{i}^{\ast}=(X_{1i}^{\ast},\ldots,X_{qi}^{\ast})^{\trans}$ and $\bm{\delta}_{i}^{\ast}=(\delta_{1i}^{\ast},\ldots,\delta_{qi}^{\ast})^{\trans}$, where $X_{ki}^{\ast}=\max(L_i,\min(T_{ki}^{\ast},U_i))$ and $\delta^*_{ki}=1$ if $X^*_{ki}= T^*_{ki}$; $\delta^*_{ki}= 2$ if $X^*_{ki}=U_{i}$; $\delta^*_{ki}= 3$ if $X^*_{ki}=L_{i}$ for $k = 1, \ldots, q$. There is a $p$-dimensional vector of baseline covariates $\mathbf{Z}_{i}$. Moreover, we assume that $(L_i, U_i)$ are independent of $(T_i,\mathbf{T}_{i}^{\ast}, \mathbf{Z}_{i})$.

Suppose that there are $N$ individuals in the EHR cohort. The true outcome $\{X_i, \delta_i\}$ is observed only for a small set of labeled data of size $n$. Without loss of generality, we assume the labeled data consists of the first $n$ individuals in the cohort. Thus, for labeled data, we observe
\[
\mathcal{D}_{l}=\left\{\big(X_i,\delta_i,L_i,U_i,\mathbf{X}_{i}^{\ast\trans},\bm{\delta}_{i}^{\ast\trans}, \mathbf{Z}_{i}^{\trans}\big)^{\trans},\,\,i=1,\ldots,n\right\},
\]
and for unlabeled data, we observe
\[
\mathcal{D}_{u}=\left\{\big(L_i,U_i,\mathbf{X}_{i}^{\ast\trans},\bm{\delta}_{i}^{\ast\trans},  \mathbf{Z}_{i}^{\trans}\big)^{\trans},\,\,i=n+1,\ldots, n+N\right\}.
\]
Furthermore, we assume that $\rho=n/(n+N) \to 0$ as $n\to \infty$, which means the missing rate of the true outcome $\{X_i, \delta_i\}$ is close to 1. The standard methods for missing data can not be applied here.

We consider the semiparametric linear transformation model due to its generality. The model assumes
\begin{equation}\label{model}
h(T)=-\Beta'\mathbf{Z}+\varepsilon,
\end{equation}
where $h$ is an unknown increasing function, $\varepsilon$ is a random variable with a known distribution which is independent of $\mathbf{Z}$, and $\Beta$ is an unknown $p$-dimensional covariate coefficient of interest. Model (\ref{model}) includes the proportional hazards model and the proportional odds model as special cases with $\varepsilon$ following the extreme value distribution and the standard logistic distribution, respectively \citep{chen2002semiparametric}.

Our goal is to develop a semi-supervised estimator of the risk effect $\Beta$ using both the labeled data $\mathcal{D}_{l}$ and the unlabeled data $\mathcal{D}_{u}$ to improve the supervised estimator that uses only $\mathcal{D}_{l}$. There are two main steps. First we obtain a supervised estimator of $\Beta$ based on $\mathcal{D}_{l}$. In the second step we augment the supervised estimator using information from $\mathcal{D}_{u}$ with the surrogate event time under some working model. The details are as follows.

\subsection{Supervised estimation based on labeled data}\label{sl}

Based on the labeled data $\mathcal{D}_l$, we adopt the nonparametric maximum likelihood approach of \citet{li2018class} for doubly censored data to estimate the risk effect $\Beta$. We first notice that  model (\ref{model}) corresponds to the semiparametric transformation models considered in \citet{li2018class} by taking $F_\epsilon(x)=1-\exp\{-G(\exp(x))\}$ and $h(t)=\log(\Lambda(t))$, where $G$ is the known increasing function, $\Lambda(t)$ denotes the unknown increasing baseline cumulative hazard function in the semiparametric transformation models. Here is an overview of the nonparametric maximum likelihood estimation through expectation–maximization (EM) algorithm with the use of subject-specific independent Poisson variables \citep{li2018class}. Suppose there exists a density $\phi(\mu\mid r)$ such that the transformation function $G$ satisfies
$
\exp\{-G(x)\} = \int_0^\infty e^{-\mu x}\phi(\mu\mid r)\,\mathrm{d}\mu.
$
To facilitate the estimation, the baseline cumulative hazard function $\Lambda(t)$ is approximated by a step function with nonnegative jumps $\lambda_k$ only at the distinct uncensored event times $t_1<\dots<t_{K_n}$.
 Introducing latent frailty variables $\mu_i\sim\phi(\mu\mid r)$ and subject-specific independent Poisson variables $N_{ik}$ with mean $\lambda_k e^{\mathbf{Z}_i^\top\Beta}\mu_i$, the estimation proceeds via an EM algorithm.

 In the E-step, the complete-data log-likelihood leads to the following $Q$-function :
\[
Q(\bm{\zeta},\bm{\zeta}^{(m)}) = \sum_{i=1}^{n}\sum_{k=1}^{K_n}\Bigl\{
\mathbf{Z}_i^{\top}\bm{\beta}\; \mathrm{E}_{\bm{\zeta}^{(m)}}(N_{ik})
+ \log\lambda_k\; \mathrm{E}_{\bm{\zeta}^{(m)}}(N_{ik})
- \lambda_k e^{\mathbf{Z}_i^{\top}\bm{\beta}} \mathrm{E}_{\bm{\zeta}^{(m)}}(\mu_i) \Bigr\},
\]
where $\boldsymbol{\zeta}=(\Beta^\top,\lambda_1,\dots,\lambda_{K_n})^\top$ and the expectations are taken with respect to the latent variables conditional on the observed data and current parameter estimate $\boldsymbol{\zeta}^{(m)}$.
The conditional expectations are computed as
\[
\mathrm{E}_{\bm{\zeta}}(N_{ik}) = I(\delta_i=3)\frac{\lambda_k e^{\mathbf{Z}_i^{\top}\bm{\beta}}}{1-e^{-G(V_i)}}I(t_k\le X_i)
+ I(\delta_i=1)I(t_k=X_i)
+ \lambda_k e^{\mathbf{Z}_i^{\top}\bm{\beta}} \mathrm{E}_{\bm{\zeta}}(\mu_i)I(t_k>X_i),
\]
\[
\mathrm{E}_{\bm{\zeta}}(\mu_i) = I(\delta_i=3)\frac{1-e^{-G(V_i)}G'(V_i)}{1-e^{-G(V_i)}}
+ I(\delta_i=1)\frac{\int \mu_i^2 e^{-\mu_i V_i}\phi(\mu_i|r)\,\mathrm{d}\mu_i}{e^{-G(V_i)}G'(V_i)}
+ I(\delta_i=2)G'(V_i),
\]
where $V_i = \sum_{t_k\le X_i}\lambda_k e^{\mathbf{Z}_i^\top\Beta}$. Note that the indicator $I(\delta_i=3)$ corresponds to left censoring, $I(\delta_i=1)$ corresponds to exact observation, and $I(\delta_i=2)$ corresponds to right censoring.

In the M-step, the jump sizes are updated via the closed form
\[
\lambda_k = \frac{\sum_{i=1}^n \mathrm{E}(N_{ik})}{\sum_{i=1}^n \mathrm{E}(\mu_i)e^{\mathbf{Z}_i^\top\Beta}},\qquad k=1,\dots,K_n,
\]
and $\Beta$ is updated by solving the profile score equation
\[
\sum_{i=1}^n\sum_{k=1}^{K_n} \mathrm{E}(N_{ik})\left\{\mathbf{Z}_i - \frac{\sum_{j=1}^n \mathrm{E}(\mu_j)e^{\mathbf{Z}_j^\top\Beta}\mathbf{Z}_j}{\sum_{j=1}^n \mathrm{E}(\mu_j)e^{\mathbf{Z}_j^\top\Beta}}\right\} = 0.
\]

The resulting estimator is denoted as $\hat{\Beta}_{SL}$. Denoting its influence function as ${\bm\xi}_i$, we have that
\[
\sqrt n(\hat\Beta_{SL}-\Beta_0)=\frac{1}{\sqrt n}\sum_{i=1}^n{\bm\xi}_i+o_p(1).
\]
The explicit form of ${\bm\xi}_i$ is derived in Appendix A.

\subsection{Semi-supervised estimation using both labeled and unlabeled data}\label{ssl}
To utilize the large unlabeled data with surrogate outcomes, we assume a working model (specified later) for the surrogate survival time $T^*$ with parameters $\btheta=({\bm\gamma}^{\trans},\alpha^{\trans})^{\trans}$, where ${\bm\gamma}$ is an unknown $p$-dimensional regression parameter analogous to $\bm\beta$, and $\alpha$ is a possibly infinite-dimensional nuisance parameter. A loss function $l_\mathcal{D}(\theta)$ based on dataset $\mathcal{D}$ is constructed under the working model. We estimate ${\bm\gamma}$ by minimizing $l_{\mathcal{D}_l}(\theta)$, denoted as $\hat{{\bm\gamma}}$, and also by minimizing $l_{{\mathcal D}_l\cup {\mathcal D}_u}(\theta)$, denoted as $\bar{{\bm\gamma}}$. If we can establish the multivariate asymptotic normality:
\begin{equation}\label{multi-normal}
\sqrt n\begin{pmatrix}
\hat{\Beta}_{SL}-\Beta_0\\
\hat{{\bm\gamma}}-\bar{{\bm\gamma}}
\end{pmatrix}\Rightarrow N\left(0,\begin{pmatrix}
\Sigma & \Omega\\
\Omega^{\trans} & \Sigma_{{\bm\gamma}}
\end{pmatrix}\right),
\end{equation}
the conditional distribution of 
\(
n^{1/2}(\hat{\bm\beta}_{SL}-\bm\beta_0) - n^{1/2}\Omega \Sigma_{\gamma}(\hat{\bm\gamma}-\bar{\bm\gamma})
\)
given \((\hat{\bm\gamma}-\bar{\bm\gamma})\) is asymptotically normal with mean zero. Let
\(
n^{1/2}(\hat{\bm\beta}_{SL}-\bm\beta_0)
= n^{1/2}\Omega \Sigma_{\gamma}(\hat{\bm\gamma}-\bar{\bm\gamma}),
\)
we obtain an updated estimator of \(\bm\beta_0\) as follows,
\begin{equation}\label{aug}
\hat{\Beta}_{SSL}=\hat{\Beta}_{SL}-\hat{\Omega}\hat{\Sigma}_{{\bm\gamma}}^{-1}(\hat{{\bm\gamma}}-\bar{{\bm\gamma}}),
\end{equation} 
where $\hat{\Omega}$ and $\hat{\Sigma}_{{\bm\gamma}}$ are consistent estimates of ${\Omega}$ and ${\Sigma}_{{\bm\gamma}}$ based on the labeled data. By properties of the multivariate normal distribution, $\sqrt n(\hat{\Beta}_{SSL}-\Beta_0)$ converges in distribution to a normal random variable with mean $0$ and covariance $\Sigma-{\Omega}{\Sigma}_{{\bm\gamma}}^{-1}\Omega^{\trans}$. The semi-supervised estimator $\hat{\Beta}_{SSL}$ is always consistent and more efficient than the original SL estimator $\hat{\Beta}_{SL}$. The estimator $\hat{\Beta}_{SSL}$ can also be understood as the most efficient estimator in the class of augmented estimators, $\hat{\Beta}_{AUG}=\hat{\Beta}_{SL}+\bm\lambda (\bar{{\bm\gamma}}-\hat{{\bm\gamma}})$, as $\bm\lambda=\hat{\Omega}\hat{\Sigma}_{{\bm\gamma}}^{-1}$ minimizes the variance of $\hat{\Beta}_{AUG}$.

The choice of the working model may affect the efficiency gain of $\hat{\Beta}_{SSL}$ over $\hat{\Beta}_{SL}$. Note that $T^*$ has the same doubly censored structure as $T$, a natural working model for the surrogate survival time $T^*$ is to take the same semiparametric linear transformation model as model (\ref{model}) for $T$. That is,
\begin{equation}\label{model2}
h^*(T^*)=-{\bm\gamma}'\mathbf{Z}+\varepsilon^*,
\end{equation}
where $h^*(\cdot)$ is an unspecified smooth strictly increasing function, and $\varepsilon^*$ has the same distribution as $\varepsilon$. Hence, ${\bm\gamma}$ can be estimated using the same EM procedure as $\hat{\Beta}_{SL}$. Let $\hat{\bm\gamma}_1$ and $\bar{\bm\gamma}_1$ be the resulting estimators for ${\bm\gamma}$ based the data from ${\mathcal D}_l$ and ${\mathcal D}_l\cup {\mathcal D}_u $, respectively, and let $\hat{\Beta}_{SSL1}$ denote the corresponding augmented estimator.

Another candidate working model for the surrogate survival time $T^*$ is the more flexible transformation model
\begin{equation}\label{model3}
h^*(T^*)=-{\bm\gamma}' \mathbf{Z}+\epsilon^* ,
\end{equation}
where $h^*(\cdot)$ is an unspecified smooth strictly increasing function, ${\bm\gamma}$ is the unknown regression parameter, and $\epsilon^*$ is a random variable with unknown distribution, which may be correlated with $\varepsilon$ in model (\ref{model}). Under this model that both $h^*(\cdot)$ and the distribution of $\epsilon^*$ are unknown, only the direction of ${\bm\gamma}$ is identifiable. However, the direction of ${\bm\gamma}$ can be recovered by maximizing a simple convex loss \citep{li1989regression} under the linearity condition for the distribution of $T^*$. Noting the doubly censored structure of $T^*$, we consider a logistic model for the left censoring probability and a Cox model for the remaining right-censored part. To further simplify the likelihood, we make another working assumption that $h^*(L) = \alpha_0 + \alpha_1 H(L)$, where $H(t)$ is some pre-specified monotone transformation function such as $\log(t)$. This leads to the following convex likelihood:
\begin{align}\label{ln}
L_n(\btheta)=&\prod_{i=1}^n\left\{\frac{\exp\{\btheta\trans\mathbf{V}_{i}\}}{1+\exp\{\btheta\trans\mathbf{V}_{i}\}}\right\}^{I(\delta^*_i=3)}
\left\{\frac{1}{1+\exp\{\btheta\trans\mathbf{V}_{i}\}}\right\}^{I(\delta^*_i\neq3)}\nonumber\\
&\times\left\{\frac{\exp\{{\bm\gamma}^{\trans} \mathbf{Z}_i\}}{\sum_{j=1}^nI(\delta^*_j\neq3,\mathbf{X}_{j}^{\ast}\geq \mathbf{X}_{i}^{\ast})\exp\{{\bm\gamma}^{\trans} \mathbf{Z}_j\}}\right\}^{I(\delta^*_i=1)},
\end{align}
where $\btheta= (\alpha_0, \alpha_1, {\bm\gamma}\trans)\trans$ and $\mathbf{V}_{i}=(1, H(L_{i}), \mathbf{Z}_i\trans)\trans$. Minimizing the negative log-likelihood $l_{\mathcal{D}_l}(\theta)=-\log(L_n(\btheta))$ and $l_{{\mathcal D}_l\cup {\mathcal D}_u}(\theta)=-\log(L_{n+N}(\btheta))$, we obtain the estimators $\hat{\bm\gamma}_2$ and $\bar{\bm\gamma}_2$, respectively. Let $\hat{\Beta}_{SSL2}$ denote the corresponding augmented estimator.

\subsection{Combined SSL estimator}
To further improve efficiency, we combine the augmented terms from both working models (\ref{model2}) and (\ref{model3}). Define $\hat{\bm\gamma}_3 = (\hat{\bm\gamma}\trans_1,\hat{\bm\gamma}\trans_2)\trans$ and $\bar{\bm\gamma}_3 = (\bar{\bm\gamma}\trans_1,\bar{\bm\gamma}\trans_2)\trans$, forming $2p$-dimensional vectors. Substituting $\hat{\bm\gamma}_3$ and $\bar{\bm\gamma}_3$ into the augmentation formula (\ref{aug}) yields the combined estimator $\hat{\Beta}_{SSL3}$, which is at least as efficient as or more efficient compared to the semi-supervised estimator with one working model alone. Because $\hat{\Beta}_{SSL3}$ is the most efficient estimator in the larger class of augmented estimators, $\hat{\Beta}_{AUG}=\hat{\Beta}_{SL}+\bm\lambda (\bar{{\bm\gamma_3}}-\hat{{\bm\gamma_3}})=\hat{\Beta}_{SL}+\bm\lambda_1 (\bar{{\bm\gamma_1}}-\hat{{\bm \gamma_1}})+\bm \lambda_2 (\bar{{\bm \gamma_2}}-\hat{{\bm \gamma_2}}).$

\section{Asymptotic properties}
We now establish the asymptotic properties of the proposed augmented estimators. Note that the surrogate data (without true labels) in $\mathcal{D}_l$ and $\mathcal{D}_u$ are independent and identically distributed. We state the required condition on the surrogate data in $\mathcal{D}_l$ for convenience. Let $O_i=(\mathbf{X}_{i}^{\ast\trans},\bm{\delta}_{i}^{\ast\trans},\mathbf{Z}_i\trans,L_i,U_i)\trans$ denote the observed surrogate data for subject $i$.
\begin{enumerate}[label={(C\arabic*)}]
\item There exists a ${\bm\gamma}^*$ such that
\[
\sqrt n(\hat{{\bm\gamma}}-{{\bm\gamma}}^*)=\frac{1}{\sqrt n}\sum_{i=1}^n {\bm \eta}_i+o_p(1),
\]
where ${\bm \eta}_i$ are independent and identically distributed with mean zero and finite covariance matrix.
\end{enumerate}

\begin{theorem}\label{Thm1}
Suppose that Condition (C1) holds. Then the asymptotic normality in (\ref{multi-normal}) holds with $\Omega=(1-\rho)\cov({\bm \xi}_i, {\bm \eta}_i)$ and ${\Sigma}_{{\bm\gamma}}=(1-\rho)\cov({\bm \eta}_i)$.
\end{theorem}

Following the discussions in Section~\ref{ssl}, $\sqrt n(\hat{\Beta}_{SSL}-\Beta_0)$ converges in distribution to a normal random variable with mean $0$ and covariance $\Sigma-{\Omega}{\Sigma}_{{\bm\gamma}}^{-1}\Omega^{\trans}$.
To ensure that the estimator of this covariance matrix is positive definite, we suggest using
\[
\hat\Sigma=\frac{1}{ n}\sum_{\mathcal{D}_l} \hat{\bm \xi}_i\hat{\bm \xi}_i\trans,\quad
\hat\Sigma_{{\bm\gamma}}=\frac{N}{n(n+N)}\sum_{\mathcal{D}_l} \hat{\bm \eta}_i\hat {\bm \eta}_i\trans,\quad
\hat\Omega=\frac{N}{n(n+N)}\sum_{\mathcal{D}_l} \hat{\bm \xi}_i\hat{\bm \eta}_i\trans,
\]
where $\hat{\bm \xi}_i$ and $\hat{\bm \eta}_i$ are the estimators of ${\bm \xi}_i$ and ${\bm \eta}_i$ obtained by plugging in $\hat\Beta_{SL}$ and $\hat{{\bm\gamma}}$, respectively.

\bigskip
\noindent{\bf Corollary 1} ~When the working model for $T^*$ is the same as the model for $T$, model (\ref{model2}), the resulting estimator $\hat{\bm\gamma}_1$ satisfies Condition (C1) with $\bm\eta_i$ defined in Appendix B.

\bigskip
\noindent{\bf Corollary 2}  ~When the working model for $T^*$ is model (\ref{model3}), the resulting estimator $\hat{\bm\gamma}_2$ satisfies Condition (C1) with $\bm\eta_i$ defined in Appendix B.

\bigskip
\noindent{\bf Corollary 3} ~ The stacked vector $\hat{\bm\gamma}_3 = (\hat{\bm\gamma}_1^\top, \hat{\bm\gamma}_2^\top)^\top$ satisfies Condition (C1) with $\bm\eta_{3i}=(\bm\eta_{1i}^\top,\bm\eta_{2i}^\top)^\top$. Moreover,
the combined estimator $\hat{\Beta}_{SSL3}$ is asymptotically more efficient than $\hat{\Beta}_{SSL1}$ and $\hat{\Beta}_{SSL2}$, in the sense that $\Sigma_3 \preceq \Sigma_1$ and $\Sigma_3 \preceq \Sigma_2$, where $\Sigma_j = \Sigma - \Omega_j \Sigma_{\gamma_j}^{-1} \Omega_j^\top$ ($j=1,2,3$) are the asymptotic covariance matrices of the three augmented estimators and $\Sigma_3 \preceq \Sigma_1$ means that $\Sigma_1-\Sigma_3$ is positive definite.\bigskip

The proofs of Theorem \ref{Thm1} and the Corollaries 1-3 are in Appendix B.

\section{Simulation Study}

In this section, we evaluated the finite-sample performance of the proposed semi-supervised estimators for doubly censored survival data and compared to the supervised estimator via simulation studies.

We simulated a cohort of size $n+N$, where $n$ is the number of labeled samples and $N$ is the number of unlabeled samples, and $N=5n$. For each subject $i$, we generated a bivariate normal covariate vector $\mathbf{Z}_i \sim \mathcal{N}(0,\Sigma)$, with $\Sigma_{jj}=1$, $\Sigma_{jk}=0.3$ ($j\neq k$), and $j, k=1, 2$. We set $h(t)=h^*(t)=\log(0.5t)$. According to the transformation linear models (\ref{model}) and (\ref{model2}), the true survival time and the surrogate survival time were generated as follows
\[
T_i = 2\exp(-\Beta'\mathbf{Z}_i+\varepsilon_i), \qquad
T_i^* = 2\exp(-{\bm\gamma}'\mathbf{Z}_i+\varepsilon_i^*),
\]
with  regression coefficients $\bm\beta=(0.5,-0.3)^\top$ and $\bm\gamma=(-0.3,0.7)^\top$. 

The random errors $\varepsilon$ and $\varepsilon^*$ are generated via a Gaussian copula to induce dependence. Specifically, we first generated $(w,w^*)$ from a bivariate normal distribution with correlation $\rho=0.85$, then apply the probability integral transform $u=\Phi(w)$, $u^*=\Phi(w^*)$ to obtain $U(0,1)$ margins. Finally, using the known survival function $S_\varepsilon(x)=\exp\{-G(\exp(x))\}$ \citep{zeng2016maximum,li2018class}, inverse transform sampling yields
\[
\varepsilon_i = \log\bigl(G^{-1}(-\log u_i, r)\bigr), \qquad
\varepsilon^*_i = \log\bigl(G^{-1}(-\log u^*_i, r^*)\bigr),
\]
where $G^{-1}$ is the inverse of $G(x,r)=\log(1+rx)/r$ ($r\geq0$). Varying $r$ allows the transformation model to recover several common models. For example, $r=0$ corresponds to the proportional hazards model (extreme-value error) and $r=1$ to the proportional odds model (logistic error).

The censoring times are generated as $L_i \sim U(0.5\tau_l,\,1.5\tau_l)$ and $U_i \sim U(0.5\tau_r,\,1.5\tau_r)$, where $\tau_l$ and $\tau_r$ denote the  20th and 80th percentiles of the survival time distribution, respectively, yielding 20\% left censoring and 20\% right censoring. For unlabeled data, the true event times are not observed; only the censoring times and surrogate outcomes are retained.

Table~\ref{tab:example} presents the estimation results for regression parameters $\bm\beta$ under different combinations of $r, r^*$ and sample sizes. When $r= r^*$, $\varepsilon$ and $\varepsilon^*$ follow the same distribution while not otherwise. Here $n=200$  and $n=400$ are considered. All results are based on 500 replications. Reported quantities are the bias (Bias), sample standard error (SE), average of estimated standard errors (ESE), 95\% empirical coverage probability (CP), and relative efficiency (RE), defined as the ratio of ESE of SL to that of SSL.

We observe that all biases are small, SEs are close to ESEs, and CPs are near the nominal level 95\%. The proposed SSL estimators significantly reduce the ESE and achieve RE larger than 1 compared to the supervised estimator under both working models. In particular, the SSL estimator under the same working transformation model as the true event time, SSL1, tends to be more efficient than that under the model general transformation model, SSL2, for most cases. The combined estimator, SSL3, is the most efficient overall, which is consistent with Corollary 3. Interestingly, when $r^* \neq r$, SSL1 based on the working model with $r$ is in fact a misspecified model for $T^*$, yet it still performs remarkably well, demonstrating the robustness of SSL1, and thus proposed augmentation framework.

\begin{table}[htbp]
\centering
\caption{Estimation results for regression parameters with simulated data.}
\resizebox{\textwidth}{!}{
\label{tab:example}
\begin{tabular}{lll l |rllll |rllll}
\hline
&&&&$\beta_1$&&&&&$\beta_2$&&&&\\
$r$&$r^*$&$n$&Methods& Bias & SE & ESE & CP & RE & Bias & SE & ESE & CP & RE \\
\hline
0&0&200 & SL & .0099 & .0843 & .0844 & 95.1 & 1.0000 & -.0053 & .0796 & .0812 & 95.5 & 1.0000 \\
&&& SSL1 & .0037 & .0605 & .0600 & 93.6 & 1.4057 & -.0009 & .0601 & .0581 & 93.6 & 1.3976 \\
&&& SSL2 & .0032 & .0616 & .0602 & 93.8 & 1.4024 & .0032 & .0616 & .0602 & 93.8 & 1.3591 \\
&&& SSL3 & .0038 & .0624 & .0595 & 93.0 & 1.4191 & -.0004 & .0607 & .0577 & 93.1 & 1.4061 \\
&&400 & SL & .0008 & .0564 & .0581 & 95.2 & 1.0000 & -.0006 & .0555 & .0560 & 94.8 & 1.0000 \\
&&& SSL1 & .0018 & .0418 & .0411 & 94.7 & 1.4131 & -.0001 & .0409 & .0400 & 94.0 & 1.4006 \\
&&& SSL2 & .0018 & .0419 & .0411 & 94.6 & 1.4132 & -.0003 & .0406 & .0410 & 95.1 & 1.3655 \\
&&& SSL3 & .0018 & .0421 & .0409 & 94.6 & 1.4227 & .0001 & .0410 & .0399 & 94.1 & 1.4052 \\
\hline
0&1&200 & SL & .0052 & .0826 & .0842 & 95.2 & 1.0000 & -.0043 & .0793 & .0817 & 95.1 & 1.0000 \\
&&& SSL1 & .0037 & .0608 & .0609 & 94.7 & 1.3838 & .0039 & .0568 & .0580 & 94.9 & 1.4076 \\
&&& SSL2 & .0039 & .0615 & .0610 & 94.2 & 1.3814 & .0043 & .0575 & .0585 & 94.4 & 1.3968 \\
&&& SSL3 & .0039 & .0613 & .0603 & 94.1 & 1.3959 & .0038 & .0580 & .0575 & 94.7 & 1.4188 \\
&&400 & SL & -.0002 & .0580 & .0582 & 94.8 & 1.0000 & .0020 & .0575 & .0561 & 94.2 & 1.0000 \\
&&& SSL1 & -.0014 & .0417 & .0418 & 94.9 & 1.3916 & .0026 & .0397 & .0398 & 94.4 & 1.4096 \\
&&& SSL2 & -.0014 & .0419 & .0418 & 94.7 & 1.3915 & .0023 & .0399 & .0400 & 94.5 & 1.4014 \\
&&& SSL3 & -.0011 & .0418 & .0416 & 95.1 & 1.4002 & .0027 & .0399 & .0396 & 94.0 & 1.4159 \\
\hline
1&0&200 & SL & -.0097 & .1258 & .1282 & 94.8 & 1.0000 & .0067 & .1287 & .1269 & 94.5 & 1.0000 \\
&&& SSL1 & -.0077 & .0861 & .0855 & 94.5 & 1.4995 & .0063 & .0868 & .0869 & 95.5 & 1.4601 \\
&&& SSL2 & -.0100 & .0961 & .0942 & 93.4 & 1.3619 & .0077 & .0924 & .0967 & 95.4 & 1.3119 \\
&&& SSL3 & -.0077 & .0861 & .0851 & 93.9 & 1.5077 & .0058 & .0866 & .0865 & 95.2 & 1.4662 \\
&&400 & SL & -.0090 & .0847 & .0899 & 96.3 & 1.0000 & .0057 & .0871 & .0889 & 95.3 & 1.0000 \\
&&& SSL1 & -.0099 & .0589 & .0601 & 95.3 & 1.4973 & .0060 & .0611 & .0608 & 94.6 & 1.4622 \\
&&& SSL2 & -.0120 & .0645 & .0661 & 94.9 & 1.3597 & .0069 & .0658 & .0678 & 95.7 & 1.3109 \\
&&& SSL3 & -.0103 & .0588 & .0598 & 94.8 & 1.5036 & .0057 & .0609 & .0606 & 94.9 & 1.4655 \\
\hline
1&1&200 & SL & -.0036 & .1241 & .1290 & 96.3 & 1.0000 & .0033 & .1235 & .1269 & 95.5 & 1.0000 \\
&&& SSL1 & -.0042 & .0876 & .0855 & 93.7 & 1.5090 & .0055 & .0843 & .0842 & 94.2 & 1.5067 \\
&&& SSL2 & -.0071 & .0980 & .0959 & 93.8 & 1.3451 & .0104 & .0962 & .0961 & 94.6 & 1.3209 \\
&&& SSL3 & -.0043 & .0876 & .0851 & 93.6 & 1.5149 & .0049 & .0847 & .0839 & 94.5 & 1.5118 \\
&&400 & SL & -.0083 & .0889 & .0900 & 95.5 & 1.0000 & .0049 & .0856 & .0887 & 95.0 & 1.0000 \\
&&& SSL1 & -.0076 & .0590 & .0592 & 94.9 & 1.5206 & .0034 & .0591 & .0585 & 95.8 & 1.5174 \\
&&& SSL2 & -.0095 & .0652 & .0669 & 94.6 & 1.3456 & .0055 & .0666 & .0673 & 95.5 & 1.3187 \\
&&& SSL3 & -.0075 & .0587 & .0591 & 95.3 & 1.5241 & .0032 & .0593 & .0583 & 95.8 & 1.5200 \\
\hline
\end{tabular}}
\end{table}

\section{Real Data Analysis}
To illustrate the proposed method, we analyzed EHR data from a health system from the US to investigate risk factors for type 2 diabetes (T2D). The cohort consists of 115,236 patients with at least one ICD code \citep{Liao2010electronic, Cipparone2015inaccuracy}, including 64,808 females and 50,428 males. The event time of interest, $T$, is age at T2D onset, which is doubly censored: some patients already had T2D at enrollment (left censoring), while others had not developed T2D by their last visit (right censoring). The left-censoring time $L$ is age at first visit, and the right-censoring time $U$ is age at last recorded visit. The true T2D onset time was obtained via manual chart review for a labeled subset of 1,613 patients (915 females, 698 males). Within this labeled sample, $T$ was exactly observed for 168 (10.4\%) patients, right-censored for 1,401 (86.9\%), and left-censored for 44 (2.7\%). We constructed a surrogate outcome $T^*$ as age at the first recorded encounter with the T2D PheCode (PheCode: 250.2). This surrogate is also doubly censored by $L$ and $U$. In the full cohort, $T^*$ was exactly observed for 16,987 (14.7\%) patients, right-censored for 95,597 (83.0\%), and left-censored for 2,652 (2.3\%).

We fitted the true event time $T$ using the linear transformation model (\ref{model}) with covariates: baseline age (age at enrollment), gender, and race. Both working models (\ref{model2}) and (\ref{model3}) were fitted to the surrogate. Table~\ref{real} presents the supervised estimator (SL) estimator and the three SSL estimators. The estimates for baseline age are negative, indicating that higher enrollment age is associated with later T2D onset. Male gender and non-white race are associated with higher risk (positive and negative coefficients, respectively). All the estimates are significant. The SSL estimators achieve substantially smaller standard errors than the SL estimator, with relative efficiency the highest (as high as 1.81) for the combined estimator. These results demonstrate that SSL can effectively leverage unlabeled data to improve inference, reducing reliance on costly manual chart review.

\begin{table}[ht]
\centering
\caption{Estimates of covariate coefficient with estimated standard errors (ESE) and relative efficiency (RE) for the T2D data.}
\begin{tabular}{rrrrrrrrrr}
\hline
& $\beta_{\text{bas-age}}$ & ESE & RE & $\beta_{\text{male}}$ & ESE & RE & $\beta_{\text{white}}$ & ESE & RE \\
\hline
SL & -.0531 & .0039 & 1.0000 & .7075 & .1788 & 1.0000 & -1.4695 & .2211 & 1.0000 \\
SSL1 & -.0560 & .0023 & 1.6956 & .5741 & .0999 & 1.7897 & -1.2419 & .1376 & 1.6068 \\
SSL2 & -.0567 & .0026 & 1.5000 & .6802 & .1068 & 1.6741 & -1.2309 & .1418 & 1.5592 \\
SSL3 & -.0569 & .0022 & 1.7727 & .5989 & .0988 & 1.8097 & -1.2001 & .1320 & 1.6750 \\
\hline
\end{tabular}
\label{real}
\end{table}

\section{Discussion}
This paper addressed a key gap in semi-supervised risk prediction with doubly censored outcomes. We developed SSL estimators of risk effects that augment a supervised estimator based on labeled data only by leveraging a large unlabeled sample with surrogate outcomes. We showed that the proposed SSL estimators are consistent and achieve higher efficiency than the supervised estimator, regardless of whether the working models for the surrogate outcome are correctly specified. A combined SSL estimator is further proposed, which is at least as efficient as any single-working-model SSL estimator. The efficiency gains reduce the need for labor-intensive manual chart review and improve the feasibility of using EHR data to study disease risk factors.

The efficiency gain depends on the strength of the association between the surrogate and the true outcome. When the surrogate is strongly correlated with the true event, the SSL estimator usually achieve substantial efficiency gain. Our numerical results show meaningful efficiency improvements under both working models, with model (\ref{model2}) generally yielding larger gains compared to model (\ref{model3}), albeit at higher computational cost. The choice of working model should balance expected efficiency gain against computational resources. Although we focused on two working models in this paper, with more working models, the semi-supervised estimation procedure is the same. The adding of useful working models that further essentially improve the efficiency depends on empirical experience.

While it is possible to construct additional augmented terms from more surrogates, working models, or estimation methods and combine them to further reduce the asymptotic variance, the incremental efficiency gain may be small if the new augmented terms are highly correlated with existing ones or only weakly associated with the influence function $\bm\xi_i$. This would require perfect knowledge of the conditional expectation of $\hat{\bm\beta}_{SL}$ given the unlabeled data, which is generally unattainable. Our numerical experience suggests that a small set of judiciously chosen working models could capture most of the achievable efficiency.

In the setting we consider, the labeled subset is very small, so the label-missing rate is close to one; standard missing-data methods are not well suited. Our approach assumes that labels are missing completely at random (MCAR), motivated by the T2D study in which patients were randomly sampled for chart review. Extending the framework to accommodate other missing mechanisms, such as missing at random (MAR), is an important direction for future work. We also assume that censoring times $(L,U)$ are independent of covariates. While this is standard for doubly censored data \citep{gu1993asymptotic}, EHR settings may involve covariate-dependent censoring due to healthcare utilization patterns. Extending the method to allow for covariate-dependent censoring is warranted.

\bibliographystyle{apalike}
\bibliography{literature}

\newpage
\section*{Appendix}
\subsection*{Appendix A}
\renewcommand{\thesection}{\Alph{section}}

\noindent \textbf{Expression of \(\bm{\xi}_i\).}
Here we derive the influence function of \(\hat{\bm{\beta}}_{SL}\).  Write the complete parameter of the EM algorithm as \(\bm{\zeta} = (\bm{\beta}^{\top},\lambda_1,\dots,\lambda_{K_n})^{\top}\).  At the \(m\)th iteration the \(Q\)-function is
\[
Q(\bm{\zeta},\bm{\zeta}^{(m)}) = \sum_{i=1}^{n}\sum_{k=1}^{K_n}\Bigl\{
\mathbf{Z}_i^{\top}\bm{\beta}\; \mathrm{E}_{\bm{\zeta}^{(m)}}(N_{ik})
+ \log\lambda_k\; \mathrm{E}_{\bm{\zeta}^{(m)}}(N_{ik})
- \lambda_k e^{\mathbf{Z}_i^{\top}\bm{\beta}} \mathrm{E}_{\bm{\zeta}^{(m)}}(\mu_i) \Bigr\},
\]
where the conditional expectations are taken under the current parameter \(\bm{\zeta}^{(m)}\) and have been defined in Section~\ref{sl}. 

Given \(\bm{\beta}\), the profile estimator of \(\lambda_k\) is
\[
\lambda_k(\bm{\beta}) = \frac{\sum_{i=1}^{n} \mathrm{E}_{\bm{\zeta}_{\bm{\beta}}}(N_{ik})}{\sum_{i=1}^{n} \mathrm{E}_{\bm{\zeta}_{\bm{\beta}}}(\mu_i) e^{\mathbf{Z}_i^{\top}\bm{\beta}}},
\qquad \bm{\zeta}_{\bm{\beta}} = (\bm{\beta}^{\top}, \lambda_1(\bm{\beta}),\dots,\lambda_{K_n}(\bm{\beta}))^{\top}.
\]
Substituting this back into the \(Q\)-function yields the profiled objective
\[
\widetilde{Q}(\bm{\beta},\bm{\beta}^{(m)})
= \sum_{i=1}^{n}\sum_{k=1}^{K_n}\Bigl\{
\mathbf{Z}_i^{\top}\bm{\beta}\; \mathrm{E}_{\bm{\zeta}^{(m)}}(N_{ik})
+ \log\lambda_k(\bm{\beta})\; \mathrm{E}_{\bm{\zeta}^{(m)}}(N_{ik})
- \lambda_k(\bm{\beta}) e^{\mathbf{Z}_i^{\top}\bm{\beta}} \mathrm{E}_{\bm{\zeta}^{(m)}}(\mu_i) \Bigr\}.
\]
The estimator \(\hat{\bm{\beta}}_{SL}\) satisfies \(\frac{\partial}{\partial\bm{\beta}}\widetilde{Q}(\bm{\beta},\hat{\bm{\beta}}_{SL})\big|_{\bm{\beta}=\hat{\bm{\beta}}_{SL}} = 0\).
Define
\[
\tilde{\psi}_i(\bm{\beta},\tilde{\bm{\beta}}) = \sum_{k=1}^{K_n}\Bigl\{
\mathbf{Z}_i \, \mathrm{E}_{\bm{\zeta}_{\tilde{\bm{\beta}}}}(N_{ik})
+ \frac{\partial/\partial\bm{\beta}\,\lambda_k(\bm{\beta})}{\lambda_k(\bm{\beta})} \mathrm{E}_{\bm{\zeta}_{\tilde{\bm{\beta}}}}(N_{ik})
- \bigl(\lambda_k(\tilde{\bm{\beta}})\mathbf{Z}_i + \frac{\partial\lambda_k(\bm{\beta})}{\partial\bm{\beta}}\bigr) e^{\mathbf{Z}_i^{\top}\bm{\beta}} \mathrm{E}_{\bm{\zeta}_{\tilde{\bm{\beta}}}}(\mu_i)
\Bigr\}.
\]
Then the first‑order condition can be written as \(\sum_{i=1}^{n}\tilde{\psi}_i(\hat{\bm{\beta}}_{SL},\hat{\bm{\beta}}_{SL})=0\).

Expanding \(\sum_i \psi_i(\hat{\bm{\beta}}_{SL},\hat{\bm{\beta}}_{SL}) \) around the true value \(\bm{\beta}_0\) gives
\[
0 = \sum_{i=1}^{n}\psi_i(\bm{\beta}_0,\bm{\beta}_0)
+ \sum_{i=1}^{n}\frac{\partial}{\partial\bm{\beta}^{\top}}\psi_i(\bm{\beta}_0,\bm{\beta}_0)(\hat{\bm{\beta}}_{SL}-\bm{\beta}_0) + o_p(n^{1/2}).
\]
Set \(A_{\beta} = -\,\mathrm{E}\bigl[ \frac{\partial}{\partial\bm{\beta}^{\top}}\psi_i(\bm{\beta}_0,\bm{\beta}_0) \bigr]\).
Then
\[
\sqrt{n}(\hat{\bm{\beta}}_{SL} - \bm{\beta}_0)
= A_{\beta}^{-1}\,\frac{1}{\sqrt{n}}\sum_{i=1}^{n}\psi_i(\bm{\beta}_0,\bm{\beta}_0) + o_p(1).
\]
Hence the influence function is \(\bm{\xi}_i= A_{\beta}^{-1}\psi_i(\bm{\beta}_0,\bm{\beta}_0)\).

\subsection*{Appendix B}
\noindent{\bf Proof of Theorem~\ref{Thm1}.}
Since $\hat{{\bm\gamma}}$ and $\bar{{\bm\gamma}}$ are computed from the labeled data and all data respectively, under Condition (C1) we have
\begin{align*}
\sqrt n(\hat{{\bm\gamma}}-\bar{{\bm\gamma}})&=\sqrt n(\hat{{\bm\gamma}}-{\bm\gamma}^*)-\sqrt{\frac{n}{n+N}}\sqrt{n+N}(\bar{{\bm\gamma}}-{\bm\gamma}^*)\\
&=\frac{1-\rho}{\sqrt n}\sum_{i=1}^n\bm\eta_i-\frac{\sqrt n}{n+N}\sum_{i=n+1}^{n+N}\bm\eta_i+o_p(1).
\end{align*}
Combining with $\sqrt n(\hat\Beta_{SL}-\Beta_0)=\frac{1}{\sqrt n}\sum_{i=1}^n\bm\xi_i+o_p(1)$ and using the multivariate central limit theorem together with independence between labeled and unlabeled data, the joint asymptotic normality follows with covariance components as stated.

\bigskip
\noindent \textbf{Proof of Corollary 1.}
When $\bm\gamma$ is estimated by applying the same EM algorithm as for $\hat{\bm\beta}_{SL}$ to the surrogate observations, it is a working maximum likelihood estimator. Let $\bm\gamma^\ast$ be the true parameter value under the working model. The $Q$-function and the profiled objective $\widetilde{Q}(\bm\gamma,\bm\gamma^{(m)})$ are defined analogously to those in the derivation for $\bm{\xi}_i$, with $\bm\beta$ replaced by $\bm\gamma$, and all conditional expectations computed with the surrogate data. At convergence, the estimator $\hat{\bm\gamma}$ satisfies
\[
\frac{\partial}{\partial\bm\gamma}\widetilde{Q}(\bm\gamma,\hat{\bm\gamma})\bigg|_{\bm\gamma=\hat{\bm\gamma}} = 0.
\]

Define
\[
\tilde{\eta}_i(\bm\gamma,\tilde{\bm\gamma}) = \sum_{k=1}^{K_n} \Bigl\{
\mathbf{Z}_i \, \mathrm{E}_{\bm\theta_{\tilde{\bm\gamma}}}(N_{ik})
+ \frac{\partial/\partial\bm\gamma\,\lambda_k(\bm\gamma)}{\lambda_k(\bm\gamma)} \mathrm{E}_{\bm\theta_{\tilde{\bm\gamma}}}(N_{ik})
- \bigl(\lambda_k(\tilde{\bm\gamma})\mathbf{Z}_i + \frac{\partial\lambda_k(\bm\gamma)}{\partial\bm\gamma}\bigr) e^{\mathbf{Z}_i^\top\bm\gamma} \mathrm{E}_{\bm\theta_{\tilde{\bm\gamma}}}(\mu_i)
\Bigr\},
\]
where $\lambda_k(\bm\gamma)$ is the profiled estimator of the baseline jump. Then the first-order condition becomes $\sum_{i=1}^n \tilde{\eta}_i(\hat{\bm\gamma},\hat{\bm\gamma})=0$.

Expanding this around the true value $\bm\gamma^\ast$ and following the same steps as in the derivation of $\bm{\xi}_i$, we obtain
\[
\sqrt{n}(\hat{\bm\gamma} - \bm\gamma^\ast) = A_{\bm\gamma}^{-1}\,\frac{1}{\sqrt{n}}\sum_{i=1}^n \tilde{\eta}_i(\bm\gamma^\ast,\bm\gamma^\ast) + o_p(1),
\]
where $A_{\bm\gamma} = -\,\mathrm{E}\bigl[ \frac{\partial}{\partial\bm\gamma^\top}\tilde{\eta}_i(\bm\gamma^\ast,\bm\gamma^\ast) \bigr]$.
Therefore, Condition (C1) holds with $\bm\eta_i = A_{\bm\gamma}^{-1}\,\tilde{\eta}_i(\bm\gamma^\ast,\bm\gamma^\ast)$.

\bigskip
\noindent \textbf{Proof of Corollary 2.}
Under working model (\ref{model2}), the estimator $\hat{\bm\gamma}_2$ is obtained by minimizing the negative composite log-likelihood $-\log L_n(\bm\theta)$ based on the surrogate observations $(X_i^\ast,\delta_i^\ast)$. The log-likelihood can be decomposed as $l_n(\bm\theta) = l_{1n}(\bm\theta) + l_{2n}(\bm\gamma)$, where
\begin{align*}
l_{1n}(\bm\theta) &= \sum_{i=1}^n I(\delta_i^\ast = 3)\,\bm\theta^\top \mathbf{V}_i - \ln\bigl(1 + \exp\{\bm\theta^\top \mathbf{V}_i\}\bigr), \\
l_{2n}(\bm\gamma) &= \sum_{i=1}^n I(\delta_i^\ast = 1) \Bigl[ \bm\gamma^\top \mathbf{Z}_i - \log \sum_{j=1}^n Y_j(X_i^\ast) e^{\bm\gamma^\top \mathbf{Z}_j} \Bigr],
\end{align*}
with $\bm\theta = (\alpha_0,\alpha_1,\bm\gamma^\top)^\top$, $\mathbf{V}_i = (1, H(L_i), \mathbf{Z}_i^\top)^\top$, $N_i(t) = I(\delta_i^\ast = 1, X_i^\ast \le t)$, and $Y_i(t) = I(\delta_i^\ast \neq 3, X_i^\ast \ge t)$.

The score functions are
\begin{align*}
\dot{l}_{1n}(\bm\theta) &= \sum_{i=1}^n \Bigl[ I(\delta_i^\ast = 3)\,\mathbf{V}_i - \frac{\exp\{\bm\theta^\top \mathbf{V}_i\}}{1+\exp\{\bm\theta^\top \mathbf{V}_i\}}\mathbf{V}_i \Bigr], \\
\dot{l}_{2n}(\bm\gamma) &= \sum_{i=1}^n \int_0^\tau \bigl[ \mathbf{Z}_i - \bar{\mathbf{Z}}(t;\bm\gamma) \bigr]\,\mathrm{d}N_i(t),
\end{align*}
where $\bar{\mathbf{Z}}(t;\bm\gamma) = S^{(1)}(t;\bm\gamma) / S^{(0)}(t;\bm\gamma)$ and $S^{(k)}(t;\bm\gamma) = n^{-1}\sum_{j=1}^n Y_j(t) e^{\bm\gamma^\top \mathbf{Z}_j} \mathbf{Z}_j^{\otimes k}$. The overall score is $\dot{l}_n(\bm\theta) = \dot{l}_{1n}(\bm\theta) + J_{\bm\gamma}\dot{l}_{2n}(\bm\gamma)$ with $J_{\bm\gamma} = (0_p,0_p,I_p)^\top$.

The Hessian $\ddot{l}_n(\bm\theta)$ is negative definite, guaranteeing a unique maximizer $\hat{\bm\theta}$ of $l_n(\bm\theta)$. By standard empirical process arguments, $\hat{\bm\theta} \xrightarrow{p} \bm\theta^\ast$, the maximizer of the limiting concave function, and $\bm\gamma^\ast$ is the corresponding subvector.

Define $s^{(k)}(t)$ as the limit of $S^{(k)}(t;\bm\gamma^\ast)$ and $\bar{\mathbf{z}}(t) = s^{(1)}(t)/s^{(0)}(t)$. Using the martingale decomposition of the Cox partial score \citep{LWYY2000}, we obtain
\begin{align*}
\frac{1}{\sqrt{n}}\dot{l}_{2n}(\bm\gamma^\ast) &= \frac{1}{\sqrt{n}}\sum_{i=1}^n \int_0^\tau \bigl( \mathbf{Z}_i - \bar{\mathbf{z}}(t) \bigr)\,\mathrm{d}M_i(t) + o_p(1),
\end{align*}
where $\mathrm{d}M_i(t) = \mathrm{d}N_i(t) - Y_i(t) e^{\bm\gamma^{*\top} \mathbf{Z}_i} \mathrm{d}\Lambda(t)$, and $\Lambda(t) = \int_0^t \frac{\mathrm{d}A(u)}{s^{(0)}(u)}$ with $A(t)$ the limit of $n^{-1}\sum_i N_i(t)$.

Consequently,
\[
\frac{1}{\sqrt{n}}\dot{l}_n(\bm\theta^\ast) = \frac{1}{\sqrt{n}}\sum_{i=1}^n \widetilde{\bm\eta}_i + o_p(1),
\]
where
\[
\widetilde{\bm\eta}_i = I(\delta_i^\ast = 3)\,\mathbf{V}_i - \frac{\exp\{\bm\theta^{*\top} \mathbf{V}_i\}}{1+\exp\{\bm\theta^{*\top} \mathbf{V}_i\}}\mathbf{V}_i + J_{\bm\gamma} \int_0^\tau \bigl( \mathbf{Z}_i - \bar{\mathbf{z}}(t) \bigr)\,\mathrm{d}M_i(t).
\]

Since $\dot{l}_n(\hat{\bm\theta}) = 0$, a Taylor expansion yields
\[
-\frac{1}{\sqrt{n}}\dot{l}_n(\bm\theta^\ast) = \bigl( n^{-1}\ddot{l}_n(\bm\theta^\ast) \bigr) \sqrt{n}(\hat{\bm\theta} - \bm\theta^\ast) + o_p(1).
\]
Let $\bm\Omega$ be the limit of $n^{-1}\ddot{l}_n(\bm\theta^\ast)$. Then
\[
\sqrt{n}(\hat{\bm\theta} - \bm\theta^\ast) = -\bm\Omega^{-1} \frac{1}{\sqrt{n}}\sum_{i=1}^n \widetilde{\bm\eta}_i + o_p(1).
\]

The estimator of $\bm\gamma$ is $\hat{\bm\gamma} = \widehat{\bm\gamma}_2$, and its influence function is obtained by taking the last $p$ components of $-\bm\Omega^{-1}\widetilde{\bm\eta}_i$. That is,
\[
\sqrt{n}(\hat{\bm\gamma} - \bm\gamma^\ast) = \frac{1}{\sqrt{n}}\sum_{i=1}^n \bm\eta_{2i} + o_p(1),
\qquad
\bm\eta_{2i} = -\bigl[ \bm\Omega^{-1} \widetilde{\bm\eta}_i \bigr]_{\bm\gamma},
\]
where $[\cdot]_{\bm\gamma}$ denotes the subvector corresponding to $\bm\gamma$. Since $\mathbb{E}[\widetilde{\bm\eta}_i] = \dot{l}(\bm\theta^\ast) = 0$, we have $\mathbb{E}[\bm\eta_{2i}] = 0$ with finite covariance, which verifies Condition (C1) for $\hat{\bm\gamma}_2$.

\bigskip
\noindent{\bf Proof of Corollary 3.}
From Corollaries~1 and~2, for \(j=1,2\),
\[
\sqrt{n}(\hat{\bm\gamma}_j - \bm\gamma_j^*) = \frac{1}{\sqrt{n}}\sum_{i=1}^n \bm\eta_{ji} + o_p(1).
\]
Stacking these expansions yields
\[
\sqrt{n}(\hat{\bm\gamma}_3 - \bm\gamma_3^*) = \frac{1}{\sqrt{n}}\sum_{i=1}^n \begin{pmatrix} \bm\eta_{1i} \\ \bm\eta_{2i} \end{pmatrix} + o_p(1),
\]
so \(\hat{\bm\gamma}_3\) satisfies (C1) with \(\bm\eta_{3i}=(\bm\eta_{1i}^\top,\bm\eta_{2i}^\top)^\top\).

To compare asymptotic efficiency, recall from Theorem~\ref{Thm1} that the asymptotic covariance matrix of the  SSL estimator is \(\bm\Sigma_j = \bm\Sigma - \bm\Omega_j\bm\Sigma_{\gamma_j}^{-1}\bm\Omega_j^\top\), where \(\bm\Omega_j = (1-\rho)\mathrm{Cov}(\bm\xi_i,\bm\eta_{ji})\) and \(\bm\Sigma_{\gamma_j} = (1-\rho)\mathrm{Var}(\bm\eta_{ji})\).
Let \(\bm U_j = \mathrm{Cov}(\bm\xi_i,\bm\eta_{ji})\) and \(\bm V_j = \mathrm{Cov}(\bm\eta_{ji})\), so that \(\bm\Sigma_j = \bm\Sigma - (1-\rho)\,\bm U_j \bm V_j^{-1} \bm U_j^\top\).

For \(j=3\), since \(\bm\eta_{3i} = (\bm\eta_{1i}^\top,\bm\eta_{2i}^\top)^\top\).
Then \(\bm U_3 = [\bm U_1,\;\bm U_2]\) and
\[
\bm V_3 = \begin{pmatrix}
\bm V_1 & \bm V_{12} \\
\bm V_{21} & \bm V_2
\end{pmatrix},
\qquad \bm V_{12} = \mathrm{Cov}(\bm\eta_{1i},\bm\eta_{2i}),\quad \bm V_{21}=\bm V_{12}^\top.
\]
Provided \(\bm V_3\) is positive definite, standard block matrix inversion gives
\[
\bm V_3^{-1} =
\begin{pmatrix}
\bm V_1^{-1} + \bm V_1^{-1}\bm V_{12} \bm S^{-1} \bm V_{21}\bm V_1^{-1} & -\bm V_1^{-1}\bm V_{12} \bm S^{-1} \\
-\bm S^{-1}\bm V_{21}\bm V_1^{-1} & \bm S^{-1}
\end{pmatrix},
\]
where \(\bm S = \bm V_2 - \bm V_{21}\bm V_1^{-1}\bm V_{12}\) is the Schur complement.
Consequently,
\begin{align*}
\bm U_3 \bm V_3^{-1} \bm U_3^\top
&= [\bm U_1,\;\bm U_2] \bm V_3^{-1} [\bm U_1,\;\bm U_2]^\top \\
&= \bm U_1 \bm V_1^{-1} \bm U_1^\top
+ (\bm U_2 - \bm U_1\bm V_1^{-1}\bm V_{12})\,\bm S^{-1}\,(\bm U_2 - \bm U_1\bm V_1^{-1}\bm V_{12})^\top.
\end{align*}
Because \(\bm S\) is positive definite, the second term is positive semidefinite. Hence
\[
\bm U_3 \bm V_3^{-1} \bm U_3^\top \succeq \bm U_1 \bm V_1^{-1} \bm U_1^\top,
\]
and therefore
\[
\bm\Sigma_3 = \bm\Sigma - (1-\rho)\,\bm U_3 \bm V_3^{-1} \bm U_3^\top \preceq \bm\Sigma - (1-\rho)\,\bm U_1 \bm V_1^{-1} \bm U_1^\top = \bm\Sigma_1.
\]
The same reasoning with the roles of \(\bm\eta_1\) and \(\bm\eta_2\) exchanged gives \(\bm\Sigma_3 \preceq \bm\Sigma_2\).
This completes the proof.
\end{document}